\documentclass{amsart}

\makeatletter
\renewcommand\@biblabel[1]{#1.}
\makeatother

\usepackage{amsthm,amssymb,amsmath,calc,eucal,ifthen}
\usepackage{amscd}
\usepackage[all,arc,curve,color,frame,matrix,arrow]{xy}

\numberwithin{equation}{section}

\newtheorem{theorem}{Theorem}[section]
\newtheorem{lemma}[theorem]{Lemma}
\newtheorem{proposition}[theorem]{Proposition}

\theoremstyle{remark}

\newtheorem{example}[theorem]{Example}

\newtheoremstyle{rmdefinition}{}{}{\upshape}{}{\bfseries}{.}{ }{}

\theoremstyle{rmdefinition}
\newtheorem{definition}[theorem]{Definition}

\newcommand{\addresss}{\begin{center}\small LaGuardia Community College of The City University of New York,\\
MEC Department, 31-10 Thomson Ave.\\
Long Island City, NY 11101, U.S.A.;\\szahariev@lagcc.cuny.edu \end{center}}

\newcommand{\be}[1]{\begin{equation}\label{#1}}
\newcommand{\ee}{\end{equation}}
\newcommand{\beqa}{\begin{eqnarray}}
\newcommand{\eeqa}{\end{eqnarray}}

\newcommand{\smtwo}[2]{\mathop{\sum}\limits_{
	\mathop{}\limits^{\text{\scriptsize $#1$}}_{
	  \text{\scriptsize $#2$}}}}
\def\ss{\mathrm{s}}

\newcounter{tmpc}
\newlength{\tmplenght}
\newlength{\tmplenghta}
\newlength{\tmplenghtb}
\newlength{\tmplenghtc}

\begin{document}

\title{Curved $A_{\infty}$-algebras and gauge theory}

\author{SVETOSLAV ZAHARIEV}
\date{}

\maketitle

\addresss

\begin{abstract} We propose a general notion of algebraic gauge theory obtained via extracting the main properties of classical gauge theory. Building on a recent work on transferring curved $A_{\infty}$-structures we show that, under certain technical conditions, algebraic gauge theories can be transferred along chain contractions. Specializing to the case of the contraction from differential forms to cochains, we obtain a simplicial gauge theory on the matrix-valued simplicial cochains of a triangulated manifold. In particular, one obtains discrete notions of connection, curvature, gauge transformation and gauge invariant action.
\end{abstract}



\section{Introduction}
Lattice gauge theory has been a widely used tool in computational physics for several decades and has also been generalized in various directions in the mathematical literature, e.g. to general simplicial complexes \cite{PS}. It is however of considerable interest to develop alternative discretization schemes in which all fundamental differential geometric and algebraic structures present in classical gauge theory have direct and transparent analogues. In this paper we propose a general algebraic framework for gauge theory and apply it to obtain a simplicial gauge theory on simplicial cochains, based on the notion of curved $A_{\infty}$-algebra.

We recall that $A_{\infty}$-algebras \cite{S} are a generalization of differential graded (dg) algebras in which the associativity condition is replaced by an infinite sequence of identities involving higher ``multiplications''. Curved $A_{\infty}$-algebras (cf. \cite{GJ} and \cite{N}) are a natural generalization of both  $A_{\infty}$-algebras and curved dg algebras \cite{Po}. The latter are generalizations of dg algebras which extract the algebraic properties of connections and curvatures on vector bundles.

It was recently shown in \cite{NZ} that under certain technical conditions one can transfer along chain contractions curved dg structures to curved $A_{\infty}$-structures, building on previous transfer results (see e.g. \cite{HK} and \cite{KS}). Applying this to Dupont's contraction from differential forms to cochains \cite{Du}, one obtains transferred notions of connection and curvature on the matrix-valued simplicial cochains of a triangulated manifold. The main theme of the present paper is transferring the remaining ingredients of classical gauge theory, namely gauge transformations, inner product and gauge-invariant action, to simplicial cochains. In order to approach this problem, we introduce a general notion of {\em algebraic gauge theory} which we explain in more detail below.

In Section 2 we present background material on curved $A_{\infty}$-algebras and their morphisms and discuss gauge transformations in dg algebras. In Section 3 we define algebraic gauge theory in the particular case of curved dg structures and provide three main examples: classical gauge theory on a smooth manifold, non-commutative gauge theories arising from A. Connes' spectral triples and gauge theories defined by deformations of dg algebras. Section 4 is devoted to certain properties of special chain contractions and auxiliary results related to the homological perturbation lemma that are needed in the sequel.

In Section 5 we define an algebraic gauge theory over an arbitrary groupoid as a functor from the groupoid to a category whose objects are triples consisting of a graded vector space $V$, a preferred class of curved $A_{\infty}$-structures on $V$ and an inner product on the tensor coalgebra of $V$, and whose morphisms are tensor coalgebra maps preserving the additional structure. We then prove our main result, Theorem \ref{dgt2}, which can be formulated as follows: the algebraic gauge theory given by the deformations of a fixed dg algebra can be transferred along a special chain contraction provided that its homotopy commute with the gauge group action. Moreover, there exists a natural transformation between the initial and the transferred functor. Applying this transfer result to Dupont's contraction we obtain, in the case of a trivial bundle, a simplicial gauge theory on the matrix-valued simplicial cochains of a triangulated manifold.

{\bf Acknowledgment.} The author is grateful to Nikolay M. Nikolov for many stimulating discussions and helpful suggestions.

\section{Curved $A_{\infty}$-algebras and their morphisms}
 We work over a fixed ground field $\textbf{k}$ of characteristic 0; all chain complexes have differentials of degree +1. In what follows $V=\bigoplus_{p}V_{p}$ always stands for a graded vector space over $\textbf{k}$, its {\em suspension} $\text{s}V$ is defined by $(\text{s}V)_{p}=V_{p+1}$. We write $V^{\otimes{k}}$ for the $k$-th graded
tensor power of $V$ so that $V^{\otimes{0}}=\textbf{k}$. We denote by
$$T(sV)=k\oplus \ss V \oplus \ss V ^{\otimes2}\oplus\ldots$$
the full tensor coalgebra over the shifted space $\ss V$. We always consider $T(\ss V)$ as a graded vector space with respect to the total grading  induced by the grading on $V$ and the grading given by the number of factors in the tensor product. When applying graded maps to graded objects, we always use the {\it Koszul sign rule},
by which we mean the appearance of a sign $ (-1)^{\text{deg}(a)\,\text{deg}(b)}$
when switching two adjacent graded symbols $a$ and $b$.

\begin{definition} A {\em curved $A_{\infty}$-structure} on a graded vector space $V$ is a degree 1 graded coderivation
$D$ on $T(\ss V)$ such that $D^{2}=0$. The pair $(V,D)$ is called a {\em curved $A_{\infty}$-algebra}.
\end{definition}
The following proposition provides a useful explicit description of curved $A_{\infty}$-structures.
\begin{proposition}\label{cainfp} (cf. \cite{GJ}) There is a bijection between curved $A_{\infty}$-structures
on $V$
and collections of linear maps $\{m_{k}\}_{k=0}^{\infty}$ such that
$$
{m_{k}}:
V^{\otimes{k}}\rightarrow V
$$
is of degree $2-k$, satisfying for every $n\geqslant 0$ the identity
\be{Ainfide}
\sum_{r+s+t \, = \, n}(-1)^{r+st} \, m_{r+t+1}
\left(\textbf{1}_{V}^{\otimes{r}}\otimes m_{s} \otimes \textbf{1}_{V}^{\otimes{t}}\right) \, = \, 0 \,.
\ee

\end{proposition}

The relationship between a curved $A_{\infty}$-structure $D$ and the maps $\{m_{k}\}$ is obtained as follows. Setting $$D_{n,k} \, := \, pr_{(\ss V)^{\otimes k}} \, \circ \, D \bigl|_{(\ss V)^{\otimes n}}:(\ss V)^{\otimes n} \rightarrow (\ss V)^{\otimes k},$$
one has
$$D_{n,k} = \,\sum_{\mathop{}\limits^{r+s+t \, = \, n}_{r+1+t=k}} \, \textbf{1}_{V}^{\otimes{r}}\otimes D_{s,1} \otimes \textbf{1}_{V}^{\otimes{t}},$$
where $D_{k,1}$ coincides with $m_k$ up to a sign (see \cite[Lemma 1.3]{GJ} for details).

We shall refer to $m_{0}(1)=D1$ (where 1 denotes the unit of $\textbf{k}$) as the {\em curvature} and  to  $m_{1}$  as the {\em connection} (or covariant derivative) of a given curved $A_{\infty}$-structure.

The notion of curved $A_{\infty}$-algebra is a generalization of the notions of $A_{\infty}$-algebra introduced by Stasheff in
\cite{S} and of curved differential graded (dg) algebra introduced by Positsel'skii in \cite{Po}. More specifically, an
{\em $A_{\infty}$-algebra structure} is given by a collection of maps $\{m_{k}\}_{k=0}^{\infty}$ as in Proposition \ref{cainfp} such that $m_{0}=0$ and a {\em curved dg algebra structure} is given by a collection of maps such that  $m_{k}=0$ for $k>2$.  Thus, a curved dg structure is defined by the first four identities in  (\ref{Ainfide}) which in this case read
\be{cdg1}
m_{1} m_{0} \, = \, 0 \,,
\ee
\be{cdg2}
m_{1} m_{1} \, = \, m_{2} (m_{0} \otimes \textbf{1})-m_{2}(\textbf{1} \otimes m_{0}) \,,
\ee
\be{cdg3}
m_{1} m_{2} \, = \, m_{2}(m_{1} \otimes \textbf{1})+m_{2}(\textbf{1} \otimes m_{1}) \,,
\ee
\be{cd4}m_{2} (m_{2} \otimes \textbf{1}) \, = \, m_{2}(\textbf{1} \otimes m_{2}) \,.
\ee
We note that (\ref{cdg1}) can be interpreted as an abstract Bianchi identity,
(\ref{cdg2}) says that square of the connection $m_{1}$ equals a commutator with
the curvature, (\ref{cdg3}) expresses the fact that $m_{1}$ is a (graded) derivation
and the last identity amounts to the associativity of multiplication given by $m_{2}$. Examples of curved dg algebras arising from vector bundles and projective modules will be described in Section 3.

\begin{definition}\label{deftmo} A {\em morphism} between two curved $A_{\infty}$-algebras $(V,D)$ and $(V',D')$ is a coalgebra map
$F: T(sV)  \rightarrow T(sV')$ which is a chain map, i.e. $FD=D'F$.
\end{definition}

\begin{proposition}\label{prmor} (cf. \cite{GJ},\cite{N}) There is a bijection between morphism of curved $A_{\infty}$-algebras $F: (V, \{m^{V}_{k}\}) \rightarrow (V', \{m^{V'}_{k}\})$ and collections of linear maps $\{F_{k}\}_{k=1}^{\infty}$ with
$F_{k}:V^{\otimes{k}}\rightarrow V'$ of degree $1-k$,
satisfying
$$
F_{1} m_{0}^{V} \, = \, m_{0}^{V'}\,,
$$
and for every $n>0$ the identity
$$
\sum_{r+s+t \, = \, n}\!(-1)^{r+st} \, F_{r+t+1} \,\, (\textbf{1}^{\otimes{r}}\otimes m^V_{s} \otimes \textbf{1}^{\otimes{t}})
$$
$$
\, = \,
\hspace{-7pt}
\smtwo{1 \, \leqslant \, q \, \leqslant \, n}{i_{1}+ \ldots + i_{q} \, = \, n}
\hspace{-9pt}
(-1)^{w} \, m^{V'}_{q} \,\, (F_{i_{1}}\otimes \ldots \otimes F_{i_{q}})\,,
$$
where
$$w=\sum_{2 \leq \ell \leq q}\Bigl((1-i_{\ell})\sum_{1\leq k\leq \ell - 1}i_{k}\Bigr).$$
\end{proposition}
The relationship between a morphism $F$  and its components $\{F_k\}$ is obtained as follows. Setting $$F_{n,k} \, := \, pr_{(\ss V')^{\otimes k}}  F \bigl|_{(\ss V)^{\otimes n}}:(\ss V)^{\otimes n} \rightarrow (\ss V')^{\otimes k},$$
one has
$$F_{n,k} = \hspace{-10pt}\sum_{i_{1}+ \ldots + i_{k} \, = \, n}\hspace{-7pt} F_{i_{1},1}\otimes \ldots \otimes F_{i_{k},1}\,,$$
where $F_{k,1}$ coincides with $F_k$ up to a sign obtained by the Koszul sign rule. The components of the composition of two morphisms $F=\{F_{k}\}_{k=1}^{\infty}$ and $G=\{G_{k}\}_{k=1}^{\infty}$ are given by the formula (cf. \cite[Section 4]{N})

\be{compmor}
(G  F)_{p}=\hspace{-7pt}\smtwo{1 \, \leqslant \, q \, \leqslant \, p}{i_{1}+ \ldots + i_{q} \, = \, p}
\hspace{-9pt}
(-1)^{w} \, G_{q} \,\, (F_{i_{1}}\otimes \ldots \otimes F_{i_{q}})\, .
\ee

\begin{example}\label{gacdg} (Gauge transformations in a dg algebra) Let $A$ be a unital dg algebra with differential $d$. Every  $\gamma \in A$ of degree 1 defines a curved dg structure $D_{\gamma}$ on $A$ given by
$$m_{0}^{\gamma}(1)=d\gamma+\gamma^{2}, \quad m_{1}^{\gamma}(a)=da+ [\gamma, a],
\quad m_{2}^{\gamma}(a\otimes b)=ab, \quad a,b \in A.$$
We denote by $G_{A}$ the multiplicative group of all degree 0 invertible elements of $A$. For every $g \in G_{A}$ and $a \in A$
we set $c(g)(a)=g a g^{-1}$ and $\gamma'=g\gamma g^{-1}+gdg^{-1}$. One easily checks that every $g \in G_{A}$ defines a morphism $F^{g}$ from  $(A,D_{\gamma})$ to
$(A,D_{\gamma' })$ whose components are given by
$$ F^{g}_{1}=c(g), \quad F^{g}_{k}=0, k>1,$$
and, using (\ref{compmor}), that the assignment $g \mapsto F^{g}$ defines a representation of $G_{A}$ on $T(sA)$.

\end{example}

\section{Algebraic gauge theory}

We denote by $\text{Aut}(T(V))$ the space of the coalgebra automorphisms of $T(V)$ and
by $\mathfrak{A}_{V}$ the set of all curved dg structures on $V$. It what follows, by inner product we mean a symmetric non-degenerate bilinear form, and by unitary operator, an invertible inner product preserving map. A {\em graded} inner product on a graded vector space $V$ is an inner product on $V$ such that $V_{i}\perp V_{j}$ for $i\neq j$.

\begin{definition}\label{agt} Let $G$ be a group and let $V$ be a graded vector space. An {\em algebraic gauge theory} over $G$ with target $V$ is a triple $(\mathfrak{C},\langle \cdot,\cdot \rangle,\rho)$, where

(1) $\mathfrak{C}$ is a subset of $\mathfrak{A}_{V}$.

(2)  $\langle \cdot,\cdot \rangle$ is a graded inner product on $T(sV)$,

(3) $\rho:G \rightarrow \text{Aut}(T(sV))$ is a unitary representation of $G$ on $T(sV)$ via coalgebra maps such that for every $g \in G$ and every $D \in \mathfrak{C}$ one has $\rho(g)D\rho(g^{-1})\in \mathfrak{C}$ and
$\rho(g)1=1$.

\end{definition}

Motivated by Example \ref{gacdg} and following the standard physics terminology, we shall call the elements of $\rho(G)$ {\em gauge transformations} and the map $S:\mathfrak{C} \rightarrow  \textbf{k}$ given by $$S(D)=\langle D1, D1 \rangle, \quad D \in \mathfrak{C}$$ the {\em action functional} of the theory.

\begin{proposition} The action functional $S$ is invariant under gauge transformations.
\end{proposition}
\begin{proof} For every $g \in G$ one has
$$S(\rho(g)D\rho(g^{-1}))=\langle \rho(g)D\rho(g^{-1})1, \rho(g)D\rho(g^{-1})1 \rangle =$$
$$=\langle \rho(g)D1, \rho(g)D1 \rangle=\langle D1, D1 \rangle=S(D).$$
\end{proof}

In the remaining part of this section we shall present three examples of algebraic gauge theory. We begin with the following simple observation.
\begin{lemma}\label{norrep} Let $\alpha: G \rightarrow \text{Aut}(V)$ be a representation of a group $G$ on a inner product vector space $V$. Assume that the adjoint operator $\alpha^{*}(g)$ exists for every $g \in G$ and that the commutator $[\alpha(g_{1}),\alpha^{*}(g_{2})]$ is equal to 0 for all $g_{1},g_{2} \in G$. Then $\rho=\alpha(\alpha ^{*})^{-1}$ is a unitary representation of $G$.

\end{lemma}
\begin{proof} We check that $\rho$ is a representation
$$\rho(g_{1}g_{2})=\alpha(g_{1}g_{2})\alpha ^{*}((g_{1}g_{2})^{-1})=
\alpha(g_{1})\alpha(g_{2})\alpha^{*}(g_{1}^{-1})\alpha^{*}(g_{2}^{-1})=$$
$$=\alpha(g_{1})\alpha^{*}(g_{1}^{-1})\alpha(g_{2})\alpha^{*}(g_{2}^{-1})=\rho(g_{1})\rho(g_{2}),$$
and that $\rho$ is unitary
$$\rho^{*}=((\alpha ^{*})^{-1})^{*}=\alpha ^{-1}\alpha ^{*}=\alpha ^{*}\alpha ^{-1}=\rho ^{-1}.$$

\end{proof}

\subsection{Classical gauge theory}
Our first example is classical gauge theory. Let $E$ be a real or complex smooth vector bundle over a Riemannian manifold $M$. Denote by $\Omega^{\bullet}_{c}(M,E)$ the algebra of compactly supported differential forms on $M$ with values in $E$,
by $\Omega_{E}:=\Omega^{\bullet}_{c}(M,\text{End}(E))$ the algebra of compactly supported forms with values in the endomorphism bundle $\text{End}(E)$ and by $\mathfrak{A}_{E}$ the set of all connections on $E$.
For every $$\mathfrak{A}_{E} \ni \nabla : \Omega^{\bullet}_{c}(M,E) \rightarrow \Omega^{\bullet +1}_{c}(M,E)$$ one can define the induced connection  $$\widetilde{\nabla} : \Omega^{\bullet}_{c}(M,\text{End}(E)) \rightarrow \Omega^{\bullet +1}_{c}(M,\text{End}(E))$$ on $\text{End}(E)$.

We define a curved dg structure $D_{\nabla}=(m_{0}^{\nabla}, m_{1}^{\nabla},m_{2}^{\nabla})$ on $\Omega_{E}$ by setting
 $$m_{0}^{\nabla}=\nabla ^{2}, \quad m_{1}^{\nabla}=\widetilde{\nabla}$$
 and taking $m_{2}^{\nabla}$ to be the composition product on $\Omega_{E}$. Using the Riemannian structure on $M$, we define an inner product on $\Omega_{E}$ by
$$\langle \omega _{1}, \omega _{2} \rangle _{E}=\int _{M} \text{Tr}(\omega _{1}\wedge \ast \omega _{2}),
\quad \omega _{1}, \omega _{2} \in \Omega_{E}.$$

The latter induces a graded inner product on the tensor coalgebra $T(s\Omega)$ which we also denote by
$\langle \cdot, \cdot \rangle_{E}$.

Let $G_{E}$ be a subgroup of the sections of the automorphism bundle $\text{Aut}(E)$. We denote by
$\rho: G_{E} \rightarrow \text{Aut}(\Omega_{E})$ the representation of $G_{E}$ on $\Omega_{E}$ given by conjugation. The latter induces a representation $\rho_{E}: G_{E} \rightarrow \text{Aut}(T(s\Omega_{E}))$.

\begin{proposition}\label{cgt}  The triple  $(\{D_{\nabla}\}_{\nabla \in \mathfrak{A}_{E}},\langle \cdot ,\cdot \rangle_{E},\rho_{E})$ defines an algebraic gauge theory over $G_{E}$ with target $\Omega_{E}$.
\end{proposition}
\begin{proof}
We observe that the operators of left multiplication and right multiplication by an element of $G_{E}$ are adjoint and that the representation $\rho$ is given by the product of the left multiplication and the inverse of the right multiplication. Hence $\rho$ is unitary by Lemma \ref{norrep} and the same is true for $\rho_{E}$. For $g \in G_{E}$ we denote by $\nabla^{g}$ the usual action of $G_{E}$ on connections. A straightforward computation shows that $\rho_{E}(g)D_{\nabla} \rho_{E}(g^{-1})=D_{\nabla^{g}}$.
\end{proof}

\subsection{Spectral triples}
Our second example is based on A. Connes' spectral triples \cite{Co}. Let $\mathcal{A}$ be an associative algebra and let $\mathcal{M}$ be a finitely generated projective left $\mathcal{A}$-module.
Let $\Omega^{\bullet}$ be a differential graded algebra such that $\Omega^{0}$ is isomorphic to $\mathcal{A}$; thus $\Omega^{\bullet}$ is an $\mathcal{A}$-bimodule.
Using this data, one defines connections (and their curvatures) on $\mathcal{M}$ in the standard fashion (see e.g. \cite{DV}.) More specifically, a {\em connection} on the pair $(\mathcal{M}, \Omega^{\bullet})$ is a degree 1 map $$\nabla: \Omega^{\bullet}\otimes_{\mathcal{A}}\mathcal{M} \rightarrow
\Omega^{\bullet}\otimes_{\mathcal{A}}\mathcal{M}$$ that satisfies the graded Leibniz rule. As in the commutative case, a connection $\nabla$ on $(\mathcal{M}, \Omega^{\bullet})$ induces a connection on
$(\text{End}_{\mathcal{A}}(\mathcal{M}),\Omega^{\bullet})$ and hence a curved dg stucture $D_{\nabla}$ on
$\Omega^{\bullet}\otimes _{\mathcal{A}} \text{End}_{\mathcal{A}}(\mathcal{M})$.

Let $(\mathcal{A},\mathcal{H}, \mathcal{D})$ be an odd spectral triple.  Thus $\mathcal{A}$ is a unital complex algebra with involution faithfully represented on a Hilbert space $\mathcal{H}$ via an involution preserving map $\phi: \mathcal{A} \rightarrow \text{End}(\mathcal{H})$, where $\text{End}(\mathcal{H})$ denotes the algebra of bounded operators on $\mathcal{H}$, and $\mathcal{D}$ is an unbounded densely defined self-adjoint operator on $\mathcal{H}$ with compact resolvent such that
$[\mathcal{D},\phi(a)]$ is bounded for every $a \in \mathcal{A}$. We suppose that $(\mathcal{A},\mathcal{H}, \mathcal{D})$ satisfies the summability and regularity conditions stated in \cite{CM} so that the algebra of pseudodifferential operators $\Psi^{\bullet}(\mathcal{A})$ associated to $(\mathcal{A},\mathcal{H}, \mathcal{D})$ and a residue trace $\tau_{\mathcal{A}}$ on $\Psi^{\bullet}(\mathcal{A})$ exist (cf. \cite[Proposition II.1]{CM}).

We assume without loss of generality that $\mathcal{D}$ is invertible, set $F=\mathcal{D}|\mathcal{D}|^{-1}$ and denote by
$\Omega^{\bullet}(\mathcal{H}, F)$ the dg algebra generated by the Fredholm module $(\mathcal{H}, F)$ (cf. \cite[IV.1.$\alpha$]{Co}). By construction $\Omega^{\bullet}(\mathcal{H}, F)$ is a subalgebra of $\text{End}(\mathcal{H})$ and $\Omega^{0}(\mathcal{H}, F)$ is isomorphic to $\mathcal{A}$. We further suppose that $\Omega^{\bullet}(\mathcal{H}, F)\subset \Psi^{\bullet}(\mathcal{A})$ and that the symmetric bilinear form defined by $\tau_{\mathcal{A}}$ is non-degenerate on $\Omega^{\bullet}(\mathcal{H}, F)$.

Let
$\text{Tr}_{\mathcal{M}}: \text{End}_{\mathcal{A}}(\mathcal{M}) \rightarrow \mathcal{A}$ be the matrix trace and set

$$\tau_{\mathcal{M}}(\omega \otimes m)=\tau_{\mathcal{A}}(\omega \otimes \text{Tr}_{\mathcal{M}}m),
\quad \omega \otimes m \in \Omega^{\bullet}\otimes _{\mathcal{A}} \text{End}_{\mathcal{A}}(\mathcal{M}).$$

One easily checks that the latter defines a trace on $\Omega^{\bullet}\otimes _{\mathcal{A}} \text{End}_{\mathcal{A}}(\mathcal{M})$ and that
$$\langle x, y \rangle_{\mathcal{D}}=\tau_{\mathcal{M}}(xy), \quad x,y \in \Omega^{k}\otimes _{\mathcal{A}} \text{End}_{\mathcal{A}}(\mathcal{M})$$
defines an inner product on $\Omega^{k}\otimes _{\mathcal{A}} \text{End}_{\mathcal{A}}(\mathcal{M})$ and hence a graded inner product on $T(s\Omega^{\bullet}\otimes _{\mathcal{A}} \text{End}_{\mathcal{A}}(\mathcal{M}))$.

Let $G_{\mathcal{M}}$ be a subgroup of $\text{Aut}_{\mathcal{A}}(\mathcal{M})$. The representation of $G_{\mathcal{M}}$ on $\text{End}_{\mathcal{A}}(\mathcal{M})$ given by conjugation extends to a representation on $\Omega^{\bullet}\otimes _{\mathcal{A}} \text{End}_{\mathcal{A}}(\mathcal{M})$ and hence also to a representation $\rho_{\mathcal{M}}$ on $T(s\Omega^{\bullet}\otimes _{\mathcal{A}} \text{End}_{\mathcal{A}}(\mathcal{M}))$. Let $\mathfrak{A}_{\mathcal{M}}$ be the set of all connections on the pair $(\mathcal{M}$,$\Omega^{\bullet}(\mathcal{H}, F))$.

\begin{proposition}The triple  $(\{D_{\nabla}\}_{\nabla \in \mathfrak{A}_{\mathcal{M}}},\langle \cdot ,\cdot \rangle_{\mathcal{M}},\rho_{\mathcal{M}})$ defines an algebraic gauge theory over $G_{\mathcal{M}}$ with target $\Omega^{\bullet}\otimes\text{End}_{\mathcal{A}}(\mathcal{M})$.
\end{proposition}

\begin{proof}The proof is completely analogous to that of Proposition \ref{cgt}.
\end{proof}

\subsection{Deformations of dg algebras}
Our third example is based on curved dg structures that are deformations. Let $A$ be a unital dg algebra. Recall that by Example \ref{gacdg} every degree 1 element $\gamma$ of $A$ defines a curved dg structure $D_{\gamma}$ on $A$. Let $G_{A}$ be a subgroup of the group of all invertible degree 0 elements of $A$. Assume that $A$ is equipped with a graded inner product such that action of $G_{A}$ on $A$ given by right multiplication is adjoint to the action given by left multiplication. This inner product induces an inner product $\langle \cdot,\cdot \rangle$ on $T(sA)$. The action of $G_{A}$ on $A$ by conjugation induces a representation $\rho$ of $G_{A}$ on $T(sA)$ as already observed in Example \ref{gacdg}.

\begin{proposition}\label{dgt}The triple  $(\{D_{\gamma}\}_{\gamma \in A_{1}},\langle \cdot ,\cdot \rangle,\rho)$ defines an algebraic gauge theory over $G_{A}$ with target $A$.
\end{proposition}

\begin{proof} The proof is analogous to that of Proposition \ref{cgt}. The fact that
$\rho(g)D_{\gamma}=D_{\gamma'}\rho(g)$ for every $g \in G_{A}$ and $\gamma'=g\gamma g^{-1}+gdg^{-1}$ was already observed in Example \ref{gacdg}.
\end{proof}

\section{Special chain contractions and the perturbation lemma}
\subsection{Special chain contractions} In this subsection we discuss certain properties of chain maps transferred via special chain contractions.
\begin{definition}\label{contr} Assume that we are given two chain complexes $(C,d)$ and $(C',d')$ and chain maps $p: C \rightarrow C'$ and
$i: C' \rightarrow C$. We say that the pair $(p,i)$ is a {\em chain contraction} of $(C,d)$ to $(C',d')$ if $pi=1_{C'}$ and there exists a homotopy between $ip$ and the identity map on $C$, i.e. a map
$H:C^{\bullet}\rightarrow C^{\bullet -1}$ such that
\begin{equation}\label{homeq}
ip-1_{C}=dH+Hd.
\end{equation}
We say that  the contraction $(p,i,H)$ is  {\em special} if the following so called annihilation conditions hold.
\begin{equation}\label{ancon}
Hi=0, \quad pH=0, \quad H^{2}=0.
\end{equation}
\end{definition}
It turns out that every chain contraction may be modified to a special one \cite{LS}. In the next three lemmas we assume that we are given
special chain contractions $(p_{k},i_{k},H_{k})$  of $(C_{k},d_{k})$ to $(C'_{k},d'_{k})$, $k=1,2,3$. Given a chain map
$\phi: C_{k} \rightarrow C_{l}$, we set $$\widehat{\phi}= p_{l}\phi \,i_{k}.$$

\begin{lemma}\label{complemma} Let $\phi: C_{1} \rightarrow C_{2}$ and $\psi: C_{2} \rightarrow C_{3}$ be chain maps such that
$\phi H_{1}=H_{2} \phi$ and  $\psi H_{2}= H_{3} \psi$. Then
$$\widehat{\psi \phi}= \widehat{\psi} \widehat{\phi}.$$
\end{lemma}

\begin{proof}We compute using (\ref{homeq}) and (\ref{ancon}) as follows:
$$\widehat{\psi} \widehat{\phi}=p_{3}\psi \, i_{2} p_{2}\phi\,  i_{1} =
p_{3}\psi(1+d_{2}H_{2}+H_{2}d_{2})\phi \, i_{1}=$$
$$=p_{3}\psi \phi \, i_{1}+p_{3}\psi d_{2}\phi H_{1} i_{1}+p_{3}H_{3}\psi d_{2}\phi \, i_{1}=\widehat{\psi \phi}.$$
\end{proof}

\begin{lemma}\label{adjlem} Let $\langle\cdot , \cdot \rangle _{k}$ be an inner product on $C_{k}$, $k=1,2$. Let $\phi: C_{1} \rightarrow C_{2}$ and $\psi: C_{2} \rightarrow C_{1}$ be chain maps satisfying $\phi H_{1}=H_{2} \phi$ and  $\psi H_{2}= H_{1} \psi$ and such that $\psi$ is the adjoint of $\phi$. Then $\widehat{\psi}$ is the adjoint of $\widehat{\phi}$ with respect to the inner products $\langle i_{k}\cdot , i_{k} \cdot \rangle _{k}$ on $C'_{k}$, $k=1,2$.
\end{lemma}

\begin{proof}
For every $c_{k} \in C'_{k},k=1,2$ one has, using (\ref{homeq}) and (\ref{ancon})
$$\langle i_{2}\widehat{\phi}c_{1} , i_{2} c_{2} \rangle _{2}=\langle i_{2}p_{2}\phi \, i_{1}c_{1} , i_{2} c_{2} \rangle _{2}=
\langle (1+d_{2}H_{2}+H_{2}d_{2})\phi \, i_{1}c_{1} , i_{2} c_{2} \rangle _{2}=$$
$$=\langle \phi \, i_{1}c_{1} , i_{2} c_{2} \rangle _{2}+\langle d_{2}\phi H_{1}i_{1}c_{1}, i_{2} c_{2}\rangle _{2}+
\langle \phi H_{1}i_{1}d'_{1}c_{1}, i_{2} c_{2}\rangle _{2}
=\langle \phi \, i_{1}c_{1} , i_{2} c_{2} \rangle _{2}.$$
On the other hand, one has
$$\langle i_{1}c_{1} , i_{1}\widehat{\psi} c_{2} \rangle _{1}=\langle  i_{1} c_{1}, i_{1}p_{1}\psi  \, i_{2}c_{2}\rangle _{1}=
\langle i_{1}c_{1} , (1+d_{1}H_{1}+H_{1}d_{1})\psi \, i_{2}c_{2} \rangle _{1}=$$
$$=\langle  i_{1}c_{1} , \psi \, i_{2} c_{2} \rangle _{1}+\langle i_{1}c_{1}, d_{1}\psi H_{2} i_{2} c_{2}\rangle _{1}+
\langle i_{1}c_{1}, \psi H_{2}i_{2}d'_{2} c_{2}\rangle _{1}
=\langle i_{1}c_{1} , \psi \, i_{2} c_{2} \rangle _{1}.$$
\end{proof}

\begin{lemma}\label{unilem}  Let $\langle\cdot , \cdot \rangle _{k}$ be an inner product on $C_{k}$, $k=1,2$ and let $\phi: C_{1} \rightarrow C_{2}$ be a unitary chain map satisfying $\phi H_{1}=H_{2} \phi$. Then $\widehat{\phi}$ is unitary with respect to the inner products  $\langle i_{k}\cdot , i_{k} \cdot \rangle _{k}$ on $C'_{k}$, $k=1,2$.
\end{lemma}
\begin{proof} This follows directly from Lemmas \ref{complemma} and \ref{adjlem}.
\end{proof}

\subsection{The perturbation lemma} We recall the coalgebra homological perturbation lemma from \cite{HK} as stated in \cite{NZ}. Given a chain complex $(C,d)$, we say that a map $\delta: C^{\bullet}\rightarrow C^{\bullet +1}$ is a {\em perturbation} of $d$ if $(d+\delta)^{2}=0$.

\begin{theorem}(cf. \cite[Section 2]{HK})\label{pertlemma}
$(a)$ Let $(C_{1},d_{1})$ be a chain complex and $\delta_{1}$ be a perturbation of $d_{1}$.
Suppose that we are given a special chain contraction $(p,i,H)$ between $(C_{1},d_{1})$
and a chain complex $(C_{2},d_{2})$.
Assume further that $\textbf{1}_{C_{1}}-\delta_{1} H$ is invertible and set
$\Sigma = (\textbf{1}_{C_{1}}-\delta_{1} H)^{-1}\delta_{1}$.
Then $\delta_{2}=p\Sigma i$ is a perturbation of $d_{2}$ and the formulas
$$
\widetilde{i} \, = \, i+H\Sigma i, \quad \widetilde{p}=p+p\Sigma H, \quad \widetilde{H}=H+H\Sigma H
$$
define a special chain contraction between $(C_{1},d_{1}+\delta_{1})$ and $(C_{2}, d_{2}+\delta_{2})$.

$(b)$ Assume, in addition to the hypotheses in part $(a)$, that $(C_{1},d_{1})$
is dg coalgebra,
$\delta_{1}$ is a coderivation, and $p$ and $i$ are coalgebra homorphisms.
Then $\delta_{2}$ is a coderivation and the maps $\widetilde{p}$ and $\widetilde{i}$ are
coalgebra homomorphisms.
\end{theorem}
 We shall recall how this theorem may be used to transfer curved $A_{\infty}$-structures and morphisms in the next section.

\begin{lemma}\label{philemma} Assume that we are given two copies of the data described in Theorem \ref{pertlemma}, i.e.
a special chain contraction $(p,i,H)$ of  $(C_{1},d_{1})$ to $(C_{2},d_{2})$, a special chain contraction
$(p',i',H')$ of $(C'_{1},d'_{1})$ to $(C_{2},d'_{2})$, a perturbation $\delta_{1}$ of $d_{1}$
and a perturbation $\delta'_{1}$ of $d'_{1}$.
Let $\varphi: (C_{1},d_{1}+\delta_{1})  \rightarrow  (C'_{1},d'_{1}+\delta'_{1})$ be a chain map such that

\begin{equation}\label{Hphi}\varphi H=H' \varphi.
\end{equation}

Then $\varphi \widetilde{H}=\widetilde{H'}\varphi$.
\end{lemma}
\begin{proof} We need to show that

\begin{equation}\label{Hsigma}\varphi H\Sigma H=H'\Sigma' H' \varphi.
\end{equation}

 We can write
$$\Sigma=1+\sum_{n=1}^{\infty}(\delta_{1} H)^{n}\delta_{1}$$
which using $H^{2}=0$ implies
\begin{equation}\label{HsigmaH}
H\Sigma H=\sum_{n=2}^{\infty}(H \delta_{1} )^{n}H
\end{equation}
Now it follows from (\ref{HsigmaH}) and (\ref{Hphi}) that in order to prove (\ref{Hsigma}) it suffices to show that
\begin{equation}\label{HdeltaH}
\varphi H \delta_{1} H = H' \delta'_{1} H' \varphi.
\end{equation}
Using that $\varphi (d_{1}+\delta_{1})=(d'_{1}+\delta'_{1})\varphi$ we see that it suffices to prove that
\begin{equation}\label{dphi}
H'(\varphi d_{1} - d'_{1} \varphi)H=0
\end{equation}
Finally, we verify (\ref{dphi}) using (\ref{ancon}) and (\ref{Hphi}):
$$ H'(\varphi d_{1} - d'_{1} \varphi)H =
H'\varphi (ip -1 - H d_{1})- (i'p' -1 - H' d'_{1})\varphi H=$$
$$=H'\varphi ip - H'\varphi - H'\varphi H d_{1}-i'p'\varphi H + \varphi H + d'_{1}H'\varphi H=0.$$

\end{proof}

\section{Algebraic gauge theory over a groupoid}
\subsection{Definition and main properties}
In this section we generalize the definition of algebraic gauge theory from Section 3 by replacing the curved dg structures with curved $A_{\infty}$-structures and the group by a groupoid. We denote the set of all curved curved $A_{\infty}$-structures on $V$ by $\mathfrak{B}_{V}$.
\begin{definition}
Let $\mathbf{J}$ be the category whose objects are all triples $(V, \langle \cdot,\cdot \rangle_{T(sV)}
, \mathfrak{C}_{V})$,
where $V$ is a graded vector space, $\langle \cdot,\cdot \rangle_{T(sV)}$ is a graded inner product on $T(sV)$ and $\mathfrak{C}_{V}$ is a subset of $\mathfrak{B}_{V}$, and whose morphisms are pairs $(\phi,\psi)$ of  unit-preserving partially isometric coalgebra morphisms
$$ \phi: T(sV_{1}) \rightarrow  T(sV_{2}) ,\quad \psi: T(sV_{2}) \rightarrow  T(sV_{1}) $$
such that
$$ \phi\,\mathfrak{C}_{V_{1}}\psi \subset  \mathfrak{C}_{V_{2}}$$
with the obvious composition law.

\end{definition}

\begin{definition}\label{agtg} Let $\mathbf{G}$ be a groupoid. An {\em algebraic gauge theory} over $\mathbf{G}$ is a functor
$\mathcal{F}: \mathbf{G} \rightarrow \mathbf{J}$.
\end{definition}

If $\mathbf{G}$ is a groupoid, we denote the set of units of $\mathbf{G}$ by $\mathbf{G}^{(0)}$ and the source and target maps by $s,t:\mathbf{G}\rightarrow \mathbf{G}^{(0)}$. We denote the image of $g \in \mathbf{G}^{(0)} $  under $\mathcal{F}$  by
$(V_{g}, \langle \cdot,\cdot \rangle_{g},\mathfrak{C}_{g})$. We call the morphisms in the image of $\mathcal{F}$ {\em gauge transformations} and the collection of functionals $\{S_{g}\}_{g \in  \mathbf{G}^{(0)}}$ given by
$$S_{g}(D)=\langle D1, D1 \rangle_{g}, \quad D \in \{\mathfrak{C}_{g}\}_{g \in  \mathbf{G}^{(0)}}$$
the {\em action functional} of the theory. As in Section 3, one shows that the action functional is invariant under gauge transformations, i.e. for every $g \in \mathbf{G}$ and every $D \in \mathfrak{C}_{s(g)}$ one has

$$S_{t(g)}(\mathcal{F}(g)D\mathcal{F}^{-1}(g))=S_{s(g)}(D).$$

In this paper we discuss only examples of algebraic gauge theory over $\mathbf{G}$ in which all spaces $\{V_{g}\}_{g \in  \mathbf{G}^{(0)}}$ coincide with a fixed space $V$. In this special case the representation of the groupoid $\mathbf{G}$ given by $\mathcal{F}$ defines an equivalence relation on the collection of curved $A_{\infty}$-structures on $V$.

\subsection{Example: Transferred curved $A_{\infty}$-structures}
Let $A$ be a unital real or complex dg algebra and $G$ be a subgroup of the group of all invertible degree 0 elements of $A$. We consider a triple $(\{D_{\gamma}\}_{\nabla \in A_{1}},\langle \cdot ,\cdot \rangle,\rho)$ as in Proposition \ref{dgt}.

Let $(p,i,H)$ be a special contraction from $A$ to a chain complex $B$ and assume that the conditions stated in
\cite[Proposition 3.3]{NZ} hold with respect to a fixed norm on $A$ and the norm on $B$ induced by the inclusion $i$. Thus according to \cite[Proposition 3.3]{NZ} there exists a subset $\Gamma_{A}$ of $A_{1}$ such that for each $\gamma \in
\Gamma_{A}$ the curved dg structure $D_{\gamma}$ on $A$ may be transferred to a curved $A_{\infty}$-structure $\widetilde{D}_{\gamma}$ on $B$. Moreover, by the homological perturbation lemma, for every $\gamma \in \Gamma_{A}$ the special contraction  $(p,i,H)$ induces a special contraction  $(P_{\gamma},I_{\gamma},H_{\gamma})$ from $T(sA)$ to $T(sB)$ such that $P_{\gamma}$ and $I_{\gamma}$ are morphisms of curved $A_{\infty}$-algebras. Thus the injection $I_{\gamma}$ induces an inner product on $T(sB)$ which we denote by $\langle \cdot ,\cdot \rangle_{B,\gamma}$.

We define a groupoid $\mathbf{G}_{A}$ by setting
$$\mathbf{G}_{A}=\bigl\{(\gamma, g, \gamma ')\in \Gamma_{A}\times G \times \Gamma_{A} \, | \, g \cdot \gamma = \gamma '\bigr\},$$
$$\mathbf{G}_{A}^{(0)}=\Gamma_{A}, \quad s(\gamma, g, \gamma ')=\gamma, \quad t(\gamma, g, \gamma ')=\gamma ',$$
and using the obvious composition law given by the multiplication in $G$. Exactly as in Subsection 3.3 one verifies the following
\begin{proposition}\label{functa} The assignment
$$  \Gamma_{A} \ni \gamma \mapsto
\bigl(A,\langle \cdot ,\cdot \rangle,\{D_{\gamma}\}_{\gamma \in \mathbf{G}_{A}^{(0)}}\bigr),$$
$$ \mathbf{G}_{A}  \ni (\gamma, g, \gamma ') \mapsto \rho(g)$$
defines an algebraic gauge theory $\mathcal{F}_{A}: \mathbf{G}_{A} \rightarrow \mathbf{J}$.
\end{proposition}

Next we show how this algebraic gauge theory can be transferred along the contraction  $(p,i,H)$. For every $(\gamma, g, \gamma ') \in \mathbf{G}_{A}$ one can transfer the morphism of curved dg algebras
$$\rho(g): (A, D_{\gamma}) \rightarrow (A, D_{\gamma '})$$
to a morphism of curved $A_{\infty}$-algebras
$$\widetilde{\rho}_{\gamma,\gamma '}(g):=P_{\gamma '}\rho(g)I_{\gamma}: (B, \widetilde{D}_{\gamma}) \rightarrow (B, \widetilde{D}_{\gamma '}).$$

\begin{theorem}\label{dgt2}Assume that for every $g \in G$ one has $[H, c(g)]=0$, where $c(g)$ denotes conjugation with $g$.
Then the assignment
$$  \Gamma_{A} \ni \gamma \mapsto
\bigl(B,\langle \cdot ,\cdot \rangle_{B, \gamma},\{\widetilde{D}_{\gamma}\}_{\gamma \in \mathbf{G}_{A}^{(0)}}\bigr),$$
$$ \mathbf{G}_{A}  \ni (\gamma, g, \gamma ') \mapsto \widetilde{\rho}_{\gamma,\gamma '}(g)$$
defines an algebraic gauge theory $\mathcal{F}_{B}: \mathbf{G}_{A} \rightarrow \mathbf{J}$. Moreover, there exists a natural transformation from $\mathcal{F}_{A}$ to $\mathcal{F}_{B}$ given by the assignment
$$ \mathbf{G}_{A}^{(0)}  \ni \gamma \mapsto (P_{\gamma}, I_{\gamma})\in
\text{Hom}(\mathcal{F}_{A}(\gamma),\mathcal{F}_{B}(\gamma)).$$
\end{theorem}

\begin{proof}
Recall that the special contraction  $(p,i,H)$ induces a special contraction $(T(p),T(i),T(H))$ from $T(sA)$ to $T(sB)$. The condition $[H, c(g)]=0$ implies that $[T(H), \rho(g)]=0$. Applying Lemma \ref{philemma} to the maps $\rho(g)$, one obtains
$$H_{\gamma'} \rho(g)=\rho(g) H_{\gamma}$$ for $\gamma '=g\cdot \gamma$. Using this, Theorem \ref{pertlemma} and the unitarity of $\rho(g)$, one concludes that by virtue of Lemma \ref{unilem} the operators $\widetilde{\rho}_{\gamma,\gamma '}(g)$ are isometries. Lemma \ref{complemma} implies that

$$\widetilde{\rho}_{\gamma',\gamma ''}(g_{2})\widetilde{\rho}_{\gamma,\gamma '}(g_{1})
=\widetilde{\rho}_{\gamma,\gamma ''}(g_{2}g_{1})$$

for $\gamma '=g_{1}\cdot \gamma$ and  $\gamma ''=g_{2}\cdot \gamma '$. Thus the maps $\widetilde{\rho}_{\gamma,\gamma '}(g)$ define a representation of the groupoid $\mathbf{G}_{A}$ on the collection
$\{B\}_{\gamma \in \mathbf{G}_{A}^{(0)}}$ and hence a representation on $\{T(sB)\}_{\gamma \in \mathbf{G}_{A}^{(0)}}$.
This implies that $\widetilde{\rho}_{\gamma,\gamma '}(g)$ are isometric isomorphisms.

To prove that the assignment $$ \mathbf{G}_{A}^{(0)}  \ni \gamma \mapsto (P_{\gamma}, I_{\gamma})\in
\text{Hom}(\mathcal{F}_{A}(\gamma),\mathcal{F}_{B}(\gamma))$$
defines a natural transformation from $\mathcal{F}_{A}$ to $\mathcal{F}_{B}$ we first observe that
$$P_{\gamma}D_{\gamma}I_{\gamma}=\widetilde{D}_{\gamma}$$
and then show that the identities
$$P_{\gamma '}\rho(g)=\widetilde{\rho}_{\gamma,\gamma '}(g)P_{\gamma},$$
$$\rho(g)I_{\gamma}=I_{\gamma '}\widetilde{\rho}_{\gamma,\gamma '}(g),$$
$$P_{\gamma'}\rho(g)D_{\gamma'}\rho^{-1}(g)I_{\gamma'}=
\widetilde{\rho}_{\gamma,\gamma '}(g)P_{\gamma}D_{\gamma}I_{\gamma}\widetilde{\rho}^{\,-1}_{\gamma,\gamma '}(g)$$
hold for $\gamma '=g\cdot \gamma$. These are verified as above, using the annihilation conditions and the identities
$\widetilde{D}_{\gamma}P_{\gamma}=P_{\gamma}D_{\gamma}$ and $I_{\gamma}\widetilde{D}_{\gamma}=D_{\gamma}I_{\gamma}$.
\end{proof}

\subsection{Example: A simplicial gauge theory} Here we consider a special case of the example from the previous subsection in order to define a discretization of classical gauge theory. Let $M$ be a compact Riemannian manifold and let $K$ be a smooth triangulation of $M$. We denote by $C^{\bullet}(K)$ the space of the real or complex simplicial cochains on $K$ and by $\Omega^{\bullet}(K)$ the space of the real or complex piece-wise smooth differential forms on $K$ (cf. \cite{Du}). We further denote by $\textbf{M}_{l}$ the algebra of the real or complex $l\times l$-matrices and by $C^{\bullet}(K,\textbf{M}_{l})$ and $\Omega^{\bullet}(K,\textbf{M}_{l})$ the spaces of matrix-valued cochains and forms.

Recall that there exists a chain contraction from $\Omega^{\bullet}(K)$ to $C^{\bullet}(K)$ (cf. \cite [Theorem 2.16] {Du} and that it is proved in \cite[Section 3]{G} that this contraction is special. Clearly this contraction extends to a contraction from $\Omega^{\bullet}(K,\textbf{M}_{l})$ to $C^{\bullet}(K,\textbf{M}_{l})$.

We define an $L^{2}$-inner product on $\Omega^{\bullet}(K,\textbf{M}_{l})$ as in Subsection 3.1, fix $e>0$ and set
$$\Gamma_{e}=\bigl\{\gamma \in  \Omega^{1}(K,\textbf{M}_{l}) \, | \, \|\gamma\|\leq e   \bigr\},$$
where $\|\cdot \|$ is the norm induced by the inner product. It is shown  in \cite[Example 3.9]{NZ} that there is a fine enough subdivision $K_{e}$ of $K$ such that every $\gamma \in \Gamma_{e}$ defines a transferred curved $A_{\infty}$-structure ${D}^{e}_{\gamma}$ on
$C^{\bullet}(K_{e},\textbf{M}_{l})$ as in subsection 4.1, using Dupont's contraction.

Let $G$ be a subgroup of the group of all invertible elements of $\Omega^{0}(K,\textbf{M}_{l})$. We define a family of graded inner products $\langle \cdot ,\cdot \rangle_{e, \gamma}$
on $T(sC^{\bullet}(K_{e},\textbf{M}_{l}))$, a groupoid $\mathbf{G}_{e}$ with unit space $\Gamma_{e}$ and a representation $\rho^{e}=\{\rho^{e}_{\gamma,\gamma '}(g)\}_{(\gamma,g,\gamma ') \in \mathbf{G}_{e}}$ of $\mathbf{G}_{e}$ on $T(sC^{\bullet}(K_{e},\textbf{M}_{l}))$ exactly as in subsection 4.1.

As in  Proposition \ref{functa}, we obtain an algebraic gauge theory $\mathcal{F}_{\Omega}: \mathbf{G}_{e} \rightarrow \mathbf{J}$ given by the assignment
$$  \Gamma_{e} \ni \gamma \mapsto
\bigl(\Omega^{\bullet}(K,\textbf{M}_{l}),\langle \cdot ,\cdot \rangle,\{D_{\gamma}\}_{\gamma \in \Gamma_{e}}
\bigr),$$
$$ \mathbf{G}_{e}  \ni (\gamma, g, \gamma ') \mapsto \rho(g).$$

\begin{theorem}\label{sgt} The assignment
$$  \Gamma_{e} \ni \gamma \mapsto
\bigl(C^{\bullet}(K_{e},\textbf{M}_{l}),\langle \cdot ,\cdot \rangle_{e, \gamma},\{D^{e}_{\gamma}\}_{\gamma \in \Gamma_{e}}\bigr),$$
$$ \mathbf{G}_{e}  \ni (\gamma, g, \gamma ') \mapsto \rho^{e}_{\gamma,\gamma '}(g)$$
defines an algebraic  gauge theory $\mathcal{F}_{C}: \mathbf{G}_{e} \rightarrow \mathbf{J}$. Moreover, there exists a natural transformation from $\mathcal{F}_{\Omega}$ to $\mathcal{F}_{C}$.
\end{theorem}

\begin{proof} It follows directly from the definition of the homotopy in Dupont's contraction given in \cite[Chapter 2]{Du} that it commutes with multiplication by $0$-forms, hence it commutes with conjugations by elements of $G$ and Theorem \ref{dgt2} is applicable.
\end{proof}

\end{document}